\documentclass[journal,draftcls,onecolumn,10pt]{IEEEtran}

\makeatletter
\let\IEEEproof\proof
\let\IEEEendproof\endproof
\let\proof\@undefined
\let\endproof\@undefined
\makeatother

\usepackage{report}

\let\proof\IEEEproof
\let\endproof\IEEEendproof

\usepackage{cite}

\begin{document}

\title{\Large\bfseries Sufficient Conditions for Existence of $J_{\alpha}(X + \sqrt[\alpha]{\eta}N)$}
 
\author{
  \authorblockN{Jihad Fahs, Ibrahim Abou-Faycal} \\
  \authorblockA{Dept.\ of Elec.\ and Comp.\ Eng.,
    American University of Beirut \\
    Beirut 1107 2020, Lebanon \\
    {\tt \{jjf03, Ibrahim.Abou-Faycal\}@aub.edu.lb}}
}

\maketitle

\begin{abstract}
In his technical report~\cite[sec. 6]{barrontech}, Barron states that the de Bruijn's identity for Gaussian perturbations holds for any RV having a finite variance. In this report, we follow Barron's steps as we prove the existence of $J_{\alpha}\left(X + \sqrt[\alpha]{\eta}N\right)$, $\eta > 0$ for any Radom Variable (RV) $X \in \mathcal{L}$ where 
\begin{equation*}
 \mathcal{L} = \left\{ \text{RVs} \,\,U: \int \ln\left(1 + |U|\right)\,dF_{U}(u) \text{ is finite } \right\},
 \end{equation*} 
 and where $N \sim \mathcal{S}(\alpha;1)$ is independent of $X$, $0< \alpha <2$. 
 \end{abstract}

\section{Main Result}
According to the definition, $J_{\alpha}(X + \sqrt[\alpha]{\eta}N)$ is the derivative of the entropy with respect to the dispersion $\eta$ of the added stable variable. Therefore, the problem boils down to proving the differentiability of $h(X + \sqrt[\alpha]{\eta}N)$.

\begin{theorem}
Let $N \sim \mathcal{S}(\alpha;1)$ and let $X \in  \mathcal{L}$ independent of $N$. Then $h(X + \sqrt[\alpha]{\eta}N)$ is differentiable in $\eta > 0$.
\end{theorem}
 First let $q_{\eta}(y) = \text{E}\left[p_{\eta}(y-X)\right]$ be the PDF of $Y = X + \sqrt[\alpha]{\eta}N$ where $p_{\eta}(\cdot)$ is the density of the alpha-stable variable with dispersion $\eta$. Note that since $p_{\eta}(\cdot)$ is bounded then so is $q_{\eta}(\cdot)$ and since $X \in \mathcal{L}$ then so is $Y$. Then $h(Y)$ is finite and is defined as
\begin{equation*}
h(Y) = -\int q_{\eta}(y) \ln q_{\eta}(y)\,dy.
\end{equation*}
We list and prove next two technical lemmas.
\begin{lemma}
\label{lemt1}
\label{firstlem}
\begin{equation*}
\frac{d}{d \eta} q_{\eta}(y) = \text{E}\left[\frac{d}{d \eta} p_{\eta}(y-X)\right]
\end{equation*}
\end{lemma}
\begin{lemma}
\label{lemt2}
\begin{equation*}
\frac{d}{d \eta} h(X + \sqrt[\alpha]{\eta}N) =  -\int \frac{d}{d \eta} \left(q_{\eta}(y)\right)\,\ln q_{\eta}(y)\,dy
\end{equation*}
exists and is finite. 
\end{lemma}
\begin{proof}
We start by proving lemma~\ref{lemt1}. The interchange holds whenever $|\frac{d}{d \eta} p_{\eta}(t)|$ is bounded uniformly by an integrable function in a neighbourhood of $\eta$ by virtue of the MVT and the Lebesgue DCT.  To prove boundedness, we start by evaluating the derivative. Since
\begin{equation*}
p_{\eta} (t) = \frac{1}{\sqrt[\alpha]{\eta}}p_N\left(\frac{t}{\sqrt[\alpha]{\eta}}\right), 
\end{equation*}
then
\begin{equation*}
\frac{d}{d \eta} p_{\eta}(t) = -\frac{1}{\alpha}\frac{1}{\eta^{1+\frac{1}{\alpha}}}p_N\left(\frac{t}{\sqrt[\alpha]{\eta}}\right) -\frac{1}{\alpha}\frac{t}{\eta^{1+\frac{2}{\alpha}}} \frac{d p_{N}}{d \eta} \left(\frac{t}{\sqrt[\alpha]{\eta}}\right),
\end{equation*}
 which gives 
 \begin{equation}
 \label{deriv}
 \left|\frac{d p_{\eta}}{d \eta} (t)\right| \leq \frac{1}{\alpha}\frac{1}{\eta^{1+\frac{1}{\alpha}}}p_N\left(\frac{t}{\sqrt[\alpha]{\eta}}\right) + \frac{1}{\alpha}\frac{|t|}{\eta^{1+\frac{2}{\alpha}}}\left|\frac{d p_N}{d u}\right|_{u= \frac{t}{\sqrt[\alpha]{\eta}}}.
 \end{equation}
 For the purpose of finding the uniform bound on the derivative, we define $b$ as a positive number chosen such that $b < \eta < 2b$. Concerning the first term of the bound in~(\ref{deriv}), we consider two separate ranges of the variable $t$ to find the uniform upperbound . On compact sets, we have
 \begin{equation}
 \label{firstcom}
 \frac{1}{\alpha}\frac{1}{\eta^{1+\frac{1}{\alpha}}}p_N\left(\frac{t}{\sqrt[\alpha]{\eta}}\right) \leq \frac{1}{\alpha} \frac{1}{b^{1+\frac{1}{\alpha}}}\max_{u \in \Reals} p_N(u)
 \end{equation}
 where the maximum exists since alpha-stable variables are unimodal~\cite{nolan:2012} and thus their PDF is upperbounded. As for large values of $|t|$, we use the fact that there exists some $k > 0$ such that $p_N(t) \leq k \frac{1}{|t|^{1+\alpha}}$~\cite{nolan:2012} which gives 
 \begin{equation}
 \label{firsth}
 \frac{1}{\alpha}\frac{1}{\eta^{1+\frac{1}{\alpha}}}p_N\left(\frac{t}{\sqrt[\alpha]{\eta}}\right) \leq \frac{k}{\alpha}\frac{1}{|t|^{1+\alpha}},
 \end{equation}  
 an integrable upperbound independent of $\eta$. Equations~(\ref{firstcom}) and~(\ref{firsth}) insures that the first term of the right-hand side (RHS) of equation~(\ref{deriv}) is uniformly upperbounded by an integrable function for $b < \eta < 2b$. When it comes to the second term of the RHS of~(\ref{deriv}), we have for $n \geq 0$ (see~\cite[p.183]{kolmo}) 
\begin{equation}
\label{derivexp}
\frac{d^n p_{N}}{d u^n}\left(u\right) = \frac{(-i)^n}{2 \pi} \int \omega^n \phi_{N}(\omega) e^{-i \omega u} \,d\omega,
\end{equation}
and
\begin{equation}
\label{uppcons}
\left|\frac{d^n p_N}{d u^n}\left(u\right)\right| \leq \frac{1}{\pi \alpha}\Gamma\left(\frac{n+1}{\alpha}\right) \end{equation}
where $\phi_{N}(\omega) = e^{- |\omega|^{\alpha}}$ is the characteristic function of $\mathcal{S}(\alpha;1)$.
Hence, on compact sets, equation~(\ref{uppcons}) gives a uniform integrable upperbound on the second term of the RHS of the form
\begin{equation}
\label{secondcom}
\frac{1}{\alpha}\frac{|t|}{\eta^{1+\frac{2}{\alpha}}}\left|\frac{d p_N}{d u}\right|_{u= \frac{t}{\sqrt[\alpha]{\eta}}} \leq \frac{1}{\pi \alpha^2}\frac{|t|}{b^{1+\frac{2}{\alpha}}}\Gamma\left(\frac{2}{\alpha}\right),
\end{equation}
which is integrable and independent of $\eta$. Therefore, we only consider next the integral term in equation~(\ref{derivexp}) at large values of $u$. To this end, we make use of the results proven in Appendix II in~\cite{Fahsarxiv}. The results of this appendix implies that $\frac{d^n p_{U}}{d u^n}\left(u\right)  = O\left(\frac{1}{|u|^{n + \alpha +1}}\right)$ when $\alpha \neq 1$, $|\beta| \neq 1$. When $\alpha =1$, the symmetric alpha-stable variable is Cauchy distributed and it is clear that $\frac{d^n p_{U}}{d u^n}\left(u\right) = \Theta\left(\frac{1}{|u|^{n+2}}\right)$. 
Since $N \sim \mathcal{S}(\alpha,1)$, then for $0 < \alpha < 2$
\begin{equation*}
\left|\frac{d^n p_{N}}{d u^n}\left(u\right)\right| = \frac{1}{2 \pi} \left|T_n(-u;0)\right| \leq \frac{\kappa_n}{|u|^{n+\alpha+1}}
\end{equation*} 
and
\begin{equation}
\label{secondh}
\frac{1}{\alpha}\frac{|t|}{\eta^{1+\frac{2}{\alpha}}}\left|\frac{d p_N}{d u}\right|_{u= \frac{t}{\sqrt[\alpha]{\eta}}} \leq \frac{1}{\alpha} \frac{\kappa_1}{|t|^{1+\alpha}}
\end{equation}
is uniformly bounded at large values of $|t|$ by an integrable function. Equations~(\ref{secondcom}) and~(\ref{secondh}) imply that the second term in the RHS of equation~(\ref{deriv}) is uniformly upperbounded by an integrable function for $b < \eta < 2b$. This proves Lemma~\ref{lemt1}.

When it comes to Lemma~\ref{lemt2}, we have the following:
\begin{align}
\frac{d}{d \eta} h(Y) &= -\int \frac{d}{d \eta} \left(q_{\eta}(y)\,\ln q_{\eta}(y)\right)\,dy \label{inter111}\\
&= -\int \frac{d q_{\eta}}{d \eta}(y)\,\ln q_{\eta}(y)\,dy -\int \frac{d q_{\eta}}{d \eta}(y) \,dy \nonumber\\
&= -\int \frac{d q_{\eta}}{d \eta}(y) \,\ln q_{\eta}(y)\,dy - \frac{d}{d \eta}\int q_{\eta}(y)\,dy \label{inter222}\\
&= -\int \frac{d q_{\eta}}{d \eta}(y) \,\ln q_{\eta}(y)\,dy. \label{final111}
\end{align}
Equation~(\ref{final111}) is true since $q_\eta(y)$ is a PDF and  integrates to $1$. Next, we start by justifying equation~(\ref{inter222}). In fact, 
\begin{eqnarray*}
\left|\frac{d q_{\eta}}{d \eta} (y)\right| &=& \left|\text{E}\left[\frac{d p_{\eta}}{d \eta} (y-X)\right]\right|\\
&\leq& \text{E}\left|\frac{d p_{\eta}}{d \eta} (y-X)\right|\\
&\leq& r_{b}(y),
\end{eqnarray*}
where the first equation is due to Lemma~\ref{lemt1} and the second is justified by the fact that the absolute value function is convex. When it comes to the last equation, it has been shown in the proof of Lemma~\ref{lemt1} that $\left|\frac{d  p_{\eta}}{d \eta}(t)\right|$ is  uniformly upperbounded in a neighbourhood of $\eta$ by an integrable function $s_b(t)$. Note that the upperbound can be written as follows by virtue of equations~(\ref{deriv}), (\ref{firstcom}), (\ref{firsth}), (\ref{secondcom}) and~(\ref{secondh}):
\begin{equation}
\label{deff}
 s_b(t) = \left\{ \begin{array}{ll}
       \displaystyle  A(b) + B(b)|t| & |t| \leq t_0\\ 
       \displaystyle  C\,p_N(t)& |t| \geq t_0,
      \end{array} \right.
  \end{equation} 
  where $A(b)$, $B(b)$, $C$ and  $t_0$ are some positive values chosen in order to write the bound. Then 
\begin{equation*}
\text{E}\left|\frac{d p_{\eta}}{d \eta} (y-X)\right| \leq \text{E}\left|s_b(y-X)\right| = r_b(y),
\end{equation*}
which is integrable since $s_b(t)$ is integrable and by using Fubini's theorem. 
This completes the justification of equation~(\ref{inter222}). As for equation~(\ref{inter111}), instead of finding a uniform integrable upperbound to $\frac{d}{d \eta} \left(q_{\eta}(y)\,\ln q_{\eta}(y)\right)$, an equivalent task is to find such one to $\frac{d  q_{\eta}(y)}{d \eta}\,\ln q_{\eta}(y)$ which we show next. 
Since $p_{N}(t) = \Theta \left( \frac{1}{|t|^{\alpha + 1}} \right)$ (see for example~\cite{nolan:2012}), there exist
positive $T$ and $K$ such that $p_{N} (t)$ is greater than $K \,
\frac{1}{|t|^{\alpha + 1}}$ for some $K$ whenever $|t| \geq T$. Now let
$y > 0$ be any scalar is large enough and define $\tilde{y} > 0$ such that $\text{Pr}(|X| \leq \tilde{y}) \geq \frac{1}{2}$. Then 
\begin{align*}
 q_{\eta}(y) = \, &  \frac{1}{\sqrt[\alpha]{\eta}} \int p_{N} \left(\frac{y-u}{\sqrt[\alpha]{\eta}}\right) \,dF_X(u)\\
    \geq \, &  \frac{1}{\sqrt[\alpha]{\eta}}\int\limits_{-\tilde{y}}^{+\tilde{y}} p_{N} \left(\frac{y-u}{\sqrt[\alpha]{\eta}}\right) \, dF_X(u)  \\
    \geq \; &  \frac{1}{2 \sqrt[\alpha]{\eta}} p_{N} \left(\frac{y+\tilde{y}}{\sqrt[\alpha]{\eta}}\right) \\
     \geq \; &  \frac{1}{2 \sqrt[\alpha]{2 b}} p_{N} \left(\frac{y+\tilde{y}}{\sqrt[\alpha]{b}}\right) \\
  \geq \; & \frac{b K}{2 \sqrt[\alpha]{2}|y + \tilde{y}|^{\alpha + 1}} \\
  \geq \;&\frac{b \tilde{K}}{|y|^{\alpha+1}},
\end{align*}
where $b < \eta < 2 b$ and $\tilde{K}$ is some positive constant.
A similar derivation may be carried for the case $y \leq - T$ large enough.
Now, since at large values of $|y|$, $q_{\eta}(y) \leq 1$, then $|\ln q_{\eta}(y)| \leq \ln \left(\frac{|y|^{\alpha+1}}{b \tilde{K}}\right)$. Furthermore since $q_{\eta}(y)$ is continuous and positive, then it achieves a positive minimum on compact subsets of $\Reals$. Let $y_0 >0$ be large enough, then on $|y| \leq y_0$, we have
\begin{align}
\left|\frac{d  q_{\eta}(y)}{d \eta}\,\ln q_{\eta}(y)\right| &\leq   \max_{y \in  \Reals} r_b(y) \left|\ln\min_{|y| \leq y_0}q_{\eta}(y)\right| \label{vlarge}\\
&\leq  \max_{y \in  \Reals} s_b(y) \left|\ln\min_{|y| \leq y_0}p_{\eta}(y)\right| \label{convrel}\\
&\leq \max_{y \in  \Reals} s_b(y) \left|\ln\frac{1}{\sqrt[\alpha]{\eta}} p_N\left(\frac{y_1}{\sqrt[\alpha]{\eta}}\right)\right| \nonumber\\
&\leq \max_{y \in  \Reals} s_b(y) \left|\ln\frac{1}{\sqrt[\alpha]{2b}} p_N\left(\frac{y_1}{\sqrt[\alpha]{b}}\right)\right| < \infty \nonumber
\end{align}
which is independent of $\eta$. We choose $y_0$ large enough in order to guarantee that $\min_{|y| \leq y_0}q_{\eta}(y) \leq 1$ and that $\max_{|y| \leq y_0} \left|\ln q_{\eta}(y)\right| \leq \left|\ln\min_{|y| \leq y_0}q_{\eta}(y)\right|$. This justifies equations~(\ref{vlarge}). The same reasoning applies to the justification of equation~(\ref{convrel}) by virtue of the fact that  $\min_{|y| \leq y_0}p_{\eta}(y) \leq \min_{|y| \leq y_0}q_{\eta}(y)$ since $q_{\eta}(y) = \text{E}\left[p_{\eta}(y-X)\right]$. Now for $|y| > y_0$, we have
\begin{equation*}
\left|\frac{d  q_{\eta}(y)}{d \eta}\,\ln q_{\eta}(y)\right| \leq  r_b(y) \left(\ln \frac{|y|^{\alpha+1}}{b \tilde{K}}\right) 
\end{equation*}
which is a uniform integrable upperbound. The integrability is justified since:
\begin{align}
&\int \ln\left(1+|y|\right)r_b(y) \,dy \\
&= \int \int \ln\left(1+|y|\right) s_b(y-x)\,dF_X(x)\,dy \nonumber\\
&= \int \int \ln\left(1+|y|\right) s_b(y-x)\,dy\,dF_X(x) \label{finter}\\
&\leq \int \int \left(\ln(1+|x|) + \ln(1+|y|)\right) s_b(y)\,dy\,dF_X(x)\nonumber\\
&=S_b  \int \ln(1+|x|) dF_X(x) + L_b \nonumber\\
&< \infty \label{ff}, 
\end{align} 
where 
\begin{equation*}
S_b = \int s_b(y)\,dy < \infty,
\end{equation*}
and \begin{equation*}
L_b = \int \ln(1+|y |)s_b(y)\,dy < \infty.
\end{equation*}
Note that $S_b$ and $L_b$ are finite by virtue of~(\ref{deff}). Equation~(\ref{finter}) is due to Fubini and equation~(\ref{ff}) is justified by the fact that $X \in \mathcal{L}$. By this, equation~(\ref{inter111}) is true and Lemma~\ref{lemt2} is proven.
\end{proof}

\bibliographystyle{IEEEtran}
\bibliography{paper}
\end{document}